\RequirePackage{amsmath}
\documentclass[runningheads,11pt]{llncs}
\usepackage[T1]{fontenc}
\usepackage{geometry}

\usepackage{amsthm}
\usepackage{setspace}
\usepackage{xcolor}
\usepackage[hidelinks]{hyperref}

\usepackage{framed}
\usepackage{amsmath}
\usepackage{amssymb}
\usepackage{tikz}
\usepackage{subcaption}
\usepackage{booktabs}
\usepackage{tabularx}
\usepackage{amssymb}
 \usepackage[algoruled,lined,linesnumbered,algonl,procnumbered]{algorithm2e}
\usepackage[size=tiny,color=orange!20]{todonotes}
\usepackage{float}
\usepackage{enumerate}

\newcommand{\eref}[1]{\hyperref[#1]{E\ref*{#1}}}
\newcommand{\cref}[1]{\hyperref[#1]{C\ref*{#1}}}
\newcommand{\pref}[1]{\hyperref[#1]{P\ref*{#1}}}

\usetikzlibrary{calc}
\usetikzlibrary{arrows,shapes,decorations.pathreplacing}

\usepackage{xspace}

\theoremstyle{remark}
\newtheorem{subcase}{Case}[case]

\theoremstyle{definition}
\newtheorem{observation}{Observation}

\DontPrintSemicolon

\usetikzlibrary{arrows.meta,bending,positioning}

\let\oldnl\nl
\newcommand{\nonl}{\renewcommand{\nl}{\let\nl\oldnl}}

\title{Silent Self--Stabilising Leader Election in Programmable Matter Systems with Holes}

\author{Jérémie Chalopin\inst{1}\orcidID{0000-0002-2988-8969} \and
Shantanu Das\inst{1}\orcidID{0000-0003-4008-2445} \and
Maria Kokkou\inst{2} \thanks{Most of this work was done while the author was affiliated with Aix-Marseille University} \thanks{Corresponding author: \email{maria.kokkou@uni-paderborn.de}}\orcidID{0009-0009-8892-3494}}
\authorrunning{J. Chalopin et al.}
\institute{
Aix Marseille Univ, CNRS, LIS, Marseille, France\\
\email{\{jeremie.chalopin,shantanu.das\}@lis-lab.fr}\and
Paderborn University, Paderborn, Germany\\
\email{maria.kokkou@uni-paderborn.de}
}

\begin{document}

\maketitle

\begin{abstract}
    Leader election is a fundamental problem in distributed computing, particularly within programmable matter systems, where coordination among simple computational entities is crucial for solving complex tasks. In these systems, particles (i.e., constant-memory computational entities) operate in a regular triangular grid as described in the geometric Amoebot model. While leader election has been extensively studied in non self--stabilising settings, self--stabilising solutions remain more limited. In this work, we study the problem of self--stabilising leader election in connected (but not necessarily \emph{simply} connected) configurations. We present the first self--stabilising algorithm for connected programmable matter systems that guarantees the election of a unique leader under an unfair scheduler, for oblivious particles (i.e., particles with no persistent memory) that share a common sense of direction. Our approach leverages particle movement, a capability not previously exploited in the self-stabilising context. We show that movement in conjunction with particles sharing a sense of orientation and operating in a grid can overcome classical impossibility results for constant-memory systems established by Dolev, Gouda and Schneider (1999).
\end{abstract}

\keywords{Leader Election \and Programmable Matter \and Self-Stabilisation \and Silent \and Deterministic \and Unique Leader \and Agreement on Directions \and Holes \and Oblivious}

\section{Introduction}
    Programmable Matter involves large collections of simple computational entities, called particles, that can change their physical properties (e.g., shape) in a programmable way and need to collaboratively accomplish a given task in a geometric environment. In this work, we assume the geometric environment to be a regular triangular grid. One of the central objectives of these systems is to be able to form any desired configuration from an arbitrary initial configuration in an efficient way with respect to time, energy and computational power. Leader election can be used as an intermediate step to designing robust algorithms for more complex problems, such as the one mentioned before, by electing a particle that can coordinate the system and break symmetries. Leader election introduced in \cite{LeLann77} is a classical problem in distributed computing, often addressed under the assumption that each node has a unique identifier. In the case of unique identifiers, it is easy to see that if nodes can exchange information, the node with the smallest or the greatest identifier can be elected. In our work, as particles have constant memory, we cannot assume unique identifiers. Leader election in anonymous systems (i.e., networks where nodes do not have unique identifiers) is impossible to solve without additional assumptions due to symmetries~\cite{angluin1980local}. To overcome this constraint, one solution is to employ randomisation \cite{itai1981symmetry} or to characterise networks where the problem can be solved \cite{yamashita1988computing}. In this paper, we use a different combination of assumptions, namely common sense of direction and movement capabilities, to elect a unique leader.
    
    Even though the large scale of programmable matter systems increases the likelihood of faults occurring and makes fault tolerance even more critical, fault-tolerant approaches to problems remain limited. In particular, self--stabilisation which is a broad way to model diverse faults such as memory corruption, particle crashes and failures in communication resulting in an arbitrarily initialised configuration, has been largely overlooked in previous work. An algorithm is self--stabilising if, from any arbitrary initial configuration, every execution of the algorithm reaches a valid configuration (whose precise definition depends on the problem) in finite time and all subsequent configurations are valid. To the best of our knowledge, \cite{chalopin2024selfstabilising} gives the only self--stabilising algorithm for programmable matter. The problem considered therein is leader election in simply connected systems, under a Gouda fair scheduler \cite{Gouda2001theory}. A Gouda fair scheduler imposes that for every configuration that appears infinitely often during an execution of an algorithm, every possible successor configuration must also appear infinitely often and is the strongest kind of fairness per \cite{dubois2011taxonomy}. A particle configuration $\mathcal{P}$ occupying a subset of nodes of an infinite regular triangular grid $G_\Delta$ is said to be \emph{simply connected} if it is connected and $G_\Delta \backslash \mathcal{P}$ is also connected. In connected systems, we call connected components of $G_\Delta \backslash \mathcal{P}$ that are surrounded by particles, \emph{holes}. The aim of this paper is to give a self--stabilising approach for the leader election problem in programmable matter systems that are connected but not necessarily simply connected, complementing the results of \cite{chalopin2024selfstabilising} for simply connected systems without holes. Our method can directly be combined with any \emph{stationary} (i.e., when particles do not have movement capabilities) self--stabilising algorithm that requires a unique leader. As we consider oblivious particles without states or persistent memory, our algorithm does not have an explicit leader state but elects a leader based on local conditions that we formally define in Section \ref{sec:movement-algo}. Informally, we show that eventually there exists exactly one particle that is locally lowermost and rightmost, which implies that this particle is globally lowermost and rightmost. We define this particle to be the leader. Our algorithm can be trivially modified to include a marked leader (see Observation~\ref{obs:marked-leader}). 

    We study self--stabilising leader election in connected programmable matter systems embedded in a regular triangular grid, where each node is incident to six edges labelled from 0 to 5. When the particles are stationary, this problem is significantly more difficult for arbitrary connected configurations compared to the simply connected setting presented in \cite{chalopin2024selfstabilising}. For example, suppose that the particles do not agree on orientation and the system is not simply connected. Then it is possible to construct a cyclic configuration where all particles have the same local information (e.g., Figure~\ref{fig:same-view}). In this setting, the results of \cite{dolev1999memory} determine that silent (i.e., all particles eventually stop performing any actions, such as updating their internal variables) self--stabilising leader election cannot be solved with constant memory. In this work, the algorithm being silent implies that particles eventually stop moving.
    \begin{figure}[ht]
        \centering
        \includegraphics[scale=.45]{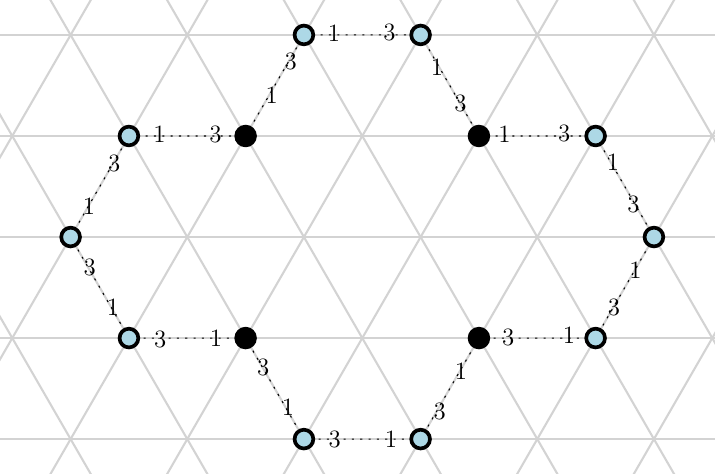}
        \caption{A particle configuration where all particles share identical local information. Blue and black particles have opposite chiralities (i.e., senses of rotational orientation). Numbers denote (consecutive) edge labels from 0 to 5.}
        \label{fig:same-view}
    \end{figure}  
    To overcome the impossibility result of \cite{dolev1999memory} we assume that particles agree on all directions. We consider the problem within the geometric Amoebot model and present a solution that uses movement, leading to a less general version of self--stabilising leader election in connected configurations. As proved in \cite{derakhshandeh2015leader}, moving cannot maintain connectivity in a self--stabilising context for programmable matter, without additional assumptions. This is due to the fact that particles may all start in an expanded state (i.e., occupying two neighbouring nodes) and immediately be instructed to contract (i.e., occupy one node), disconnecting the system with no way of reconnecting, while the system is not in a valid state. Therefore, we need to assume that particles can \textit{sense} their surroundings in a way that the information on whether neighbouring nodes are occupied by other particles or empty cannot be corrupted. 

\subsection{Related Work}
Programmable matter was introduced in \cite{toffoli1991programmable} and has since been widely studied both within a theoretical context and a practical one (e.g., \cite{piranda2018geometrical,piranda2022datom}). In this work we focus on its theoretical aspect, for which multiple models have been defined. One of the most popular active models (i.e., when the computational entities are also the building blocks of the system) is the \textit{Amoebot model} \cite{AM-derakhshandeh2014amoebot,daymude2023canonical} that has also been the basis for SILBOT~\cite{d2020asynchronous} and the Reconfigurable Circuits extension~\cite{feldmann2021accelerating}. Multiple other models exist such as \cite{almethen2020pushing,fekete2021cadbots,gmyr2017forming,woods2013active}. We focus on leader election results within Amoebot which is the model used in this work.

\begin{table*}[htp]
    \centering
    \caption{Deterministic Leader Election in triangular grids. ``Simply Connected'' denotes particle systems without \textit{holes}. ``Chirality'' is a common sense of rotational orientation. ``Movement'' is the ability of particles to move to neighbouring nodes. A ``Sequential (Seq.) Scheduler'' activates one particle at a time. An ``Asynchronous (Asynch.) Scheduler'' activates particles independently. A ``Fair'' scheduler eventually activates all activable particles. An ``Unfair'' scheduler activates a number of activable particles. Under a ``Gouda Fair'' scheduler, every successor of a recurring configuration is recurring. 
    $^{*}$Agreement on the directions along one axis.
    $^{\dag}$Agreement on all directions, which implies chirality.
    }
    \label{tab:previous-LE}
    \renewcommand{\arraystretch}{1}  
    \resizebox{\textwidth}{!}{%
    \begin{tabular}{>{\centering\arraybackslash}p{2.5cm}|
                    >{\centering\arraybackslash}p{2cm}|
                    >{\centering\arraybackslash}p{2.5cm}|
                    >{\centering\arraybackslash}p{2cm}|
                    >{\centering\arraybackslash}p{2.2cm}|
                    >{\centering\arraybackslash}p{2.7cm}|
                    >{\centering\arraybackslash}p{2.5cm}}
        \toprule
        \textbf{Paper} & 
        \textbf{Leaders} & 
        \textbf{Simply Connected} & 
        \textbf{Chirality} & 
        \textbf{Movement} & 
        \textbf{Scheduler} & 
        \textbf{Self-Stabilising} \\
        \midrule
        \cite{dufoulon2021efficient} & 6 & \texttt{X} & \checkmark & \texttt{X} & Seq. Fair & \texttt{X}\\
        \hline
        \cite{di2020shape} & 3 & \checkmark & \texttt{X} & \texttt{X} & Asynch. Fair & \texttt{X}\\ 
        \hline
        \cite{bazzi2019stationary} & 6 & \texttt{X} & \checkmark & \texttt{X} & Asynch. Fair & \texttt{X}\\  
        \hline
        \cite{gastineau2019distributed} & 1 & \checkmark & \checkmark & \texttt{X} & Seq. Fair & \texttt{X} \\ 
        \hline
        \cite{emek2019deterministic} & 1 & \texttt{X} & \texttt{X} & \checkmark & Seq. Fair & \texttt{X} \\
        \hline
        \cite{dufoulon2021efficient} & 1 & \texttt{X} & \checkmark & \checkmark & Seq. Fair & \texttt{X} \\
        \hline
        \cite{briones2023asynchronous} & 1 & \checkmark & \texttt{X} & \texttt{X} & Seq. Fair & \texttt{X}\\
        \hline
        \cite{chalopin2024deterministic} & 1 & \texttt{X} & \texttt{X}$^*$ & \texttt{X} & Asynch. Fair & \texttt{X} \\
        \hline
        \cite{chalopin2024selfstabilising} & 1 & \checkmark & \texttt{X} & \texttt{X} & Seq. Gouda Fair & \checkmark \\
        \hline
        \textbf{This Work} & 1 & \texttt{X} & \checkmark$^{\dag}$ & \checkmark & Seq. Unfair & \checkmark \\
        \toprule
    \end{tabular}
    }
  \end{table*}
Leader election is a very well studied problem within Amoebot, with results under various assumptions. It has been studied both in the three-dimensional setting (e.g., \cite{briones2023asynchronous,gastineau2022leader}) and in two dimensions. In the 2D case, both randomised (e.g., \cite{derakhshandeh2015leader}) and deterministic algorithms have been proposed. In the deterministic approach particles are often assumed to agree on rotational orientation (known as particles having common \textit{chirality}) like for example in \cite{bazzi2019stationary,dufoulon2021efficient,gastineau2019distributed}. In \cite{chalopin2024deterministic} the authors assume agreement on one axis of the triangular grid and show that this assumption is not comparable to agreement on chirality. Earlier work on leader election, as for example \cite{dufoulon2021efficient,emek2019deterministic,gastineau2019distributed}, assumes a sequential scheduler (i.e., one particle is active at any time). However, the problem has also been studied under an asynchronous scheduler in \cite{bazzi2019stationary,chalopin2024deterministic,di2020shape}. In terms of fairness, all previous leader election algorithms assume a fair scheduler (i.e., a scheduler that eventually activates every particle that can be activated). However, although this is not stated explicitly, we believe that most algorithms, especially \cite{dufoulon2021efficient,gastineau2019distributed,chalopin2024deterministic}, also work under an unfair scheduler (i.e., at any time at least one particle that
can be activated is chosen but some particles may be perpetually ignored as long as another activable particle exists in the system). Most results in the literature consider stationary particles, with the exception of \cite{dufoulon2021efficient,emek2019deterministic} which use the movement capabilities of programmable matter. 
Finally, \cite{di2020shape,gastineau2019distributed} consider simply connected particle configurations whereas \cite{bazzi2019stationary,chalopin2024deterministic,dufoulon2021efficient,emek2019deterministic} allow for the particle system to contain holes. The results for the two-dimensional deterministic case in the Amoebot model are also summarised in Table \ref{tab:previous-LE}.

In the same model, \cite{chalopin2024selfstabilising} studied self-stabilising leader election in the context of programmable matter for the first time, demonstrating that geometry can be leveraged to overcome the impossibility results established for general graphs in \cite{dolev1999memory}. In particular, \cite{chalopin2024selfstabilising} gives a deterministic self--stabilising leader election algorithm for particles with constant memory when the particle system is simply connected and particles are activated by a Gouda fair scheduler. To the best of our knowledge, \cite{chalopin2024selfstabilising} is the only paper studying self--stabilisation in the context of programmable matter. Some earlier work also discusses self--stabilisation in programmable matter but in slightly different settings. In \cite{derakhshandeh2015leader}, the authors discuss the possibility of making their randomised leader election algorithm self-stabilising by combining it with techniques from \cite{awerbuch1994memory,itkis1994fast}. However, in that case, it is assumed that particles have $O(\log^*n)$ memory, where $n$ is the number of particles in the system, bypassing the constant memory constraint of programmable matter systems. More recently, \cite{daymude2021bio} introduced a deterministic self-stabilising algorithm for constructing a spanning forest for particles with constant memory. However, in that case it is assumed that at least one non-faulty special particle always remains in the system, making the design of self--stabilising leader election algorithms particularly important. Hence the field of self--stabilisation within programmable matter systems remains largely unexplored as highlighted in \cite{daymude2023canonical}.


Within the SILBOT model, Navarra and Piselli~\cite{navarra2023asynchronous} introduced the idea of moving particles with a lower neighbour downwards until the particle system forms a single line. While their algorithm also yields a unique lowermost rightmost particle, it is not self--stabilising. Their approach assumes that every particle is initially contracted, movement is done via expansions and contractions and particles are allowed to become disconnected during execution. In a self--stabilising setting, every particle state allowed by the algorithm is a potential state of a particle in the initial configuration. Hence, in the self-stabilising context, the algorithm of \cite{navarra2023asynchronous} would allow particles to be initially expanded and for the system to initially be disconnected. If all particles in the system are initially expanded, instructing them to contract may cause disconnection. Moreover, if the system is initially disconnected, the particles have no way of knowing or reconnecting. Therefore, although we adopt the idea of moving particles with a lower neighbour down, directly using the algorithm of \cite{navarra2023asynchronous} is not possible. We address the aforementioned challenges by introducing atomic moves that consist of at most one contraction and at most one expansion, as well as a set of conditions that preserve connectivity at every step. Finally, \cite{navarra2023asynchronous} relies on additional assumptions such as two hop visibility, which we do not assume.

\subsection{Our Contributions}\label{sec:contributions}
We present a silent, self-stabilising leader election algorithm that deterministically ensures that in any arbitrarily initialised connected system, there eventually exists a unique locally lowermost rightmost particle, which is defined to be the leader. Being silent is in general a desirable property of self-stabilising algorithms as it implies smaller communication bandwidth~\cite{dolev1999memory}. In the context of our work, where particles do not communicate with each other, our algorithm being silent implies that particles eventually stop moving. This property allows our algorithm to be combined with other stationary self--stabilising algorithms. In order to overcome the additional difficulties introduced in a connected system instead of a simply connected one, we use atomic moves consisting of at most one contraction and at most one expansion. Algorithm \ref{alg:LE-movement} is the first self--stabilising algorithm that works in a connected configuration (instead of a simply connected configuration like \cite{chalopin2024selfstabilising}). Furthermore, this is the first self--stabilising algorithm for programmable matter that works under an unfair scheduler. In order to achieve this, we assume that particles agree on all directions of the grid and are activated sequentially. The latter assumption means that only one particle is active at a time, which is a standard assumption in programmable matter systems. We additionally assume that particles are oblivious, making our algorithm the first self--stabilising leader election algorithm for programmable matter where particles do not need persistent memory. Our algorithm (Section \ref{sec:movement-algo}) only depends on an active particle detecting which nodes in its neighbourhood are occupied by other particles, but particles cannot exchange any additional information. In particular, particles cannot exchange messages or read the state of neighbouring particles like in previous work within programmable matter.

\section{Model and Preliminaries}
Let $G_\Delta$ be an infinite regular triangular grid. We assume the particle system $\mathcal{P}$ occupies a connected subset of nodes in $G_\Delta$. Each computational entity in $\mathcal{P}$ is called a \textit{particle}. We assume that each particle is \textit{oblivious} (i.e., has no memory), has no communication capabilities and can be either \textit{contracted} occupying one node or \textit{expanded} occupying two neighbouring nodes. Notice that the state of the particle is thus only defined by whether it is expanded or contracted. Each node can be occupied by at most one particle at any time. We call the two endpoints of an expanded particle the \textit{head} and \textit{tail} of the particle. A contracted particle $p$ is incident to six ports, each corresponding to one of the six neighbouring nodes of $p$ in $G_\Delta$, which are arranged consecutively in cyclic order around $p$. An expanded particle is incident to eight ports, each corresponding to one of the eight nodes around it in $G_\Delta$. We define the neighbourhood of each endpoint of an expanded particle to be the same as the neighbourhood of a contracted particle, that is, to be the six nodes reachable by ports. We consider the neighbourhood of an expanded particle to be the union of the neighbourhood of the two endpoints of the particle. This means that we consider the tail (resp. head) of an expanded particle to be part of the neighbourhood of the head (resp. tail) of the particle. We write \textit{occupied neighbourhood} of a particle $p$ or $N(p)$ to denote the nodes that are neighbouring to $p$ and are occupied by particles. We use the notation $N[p] = N(p) \cup \{p\}$ to refer to the occupied neighbourhood of $p$ and the node(s) occupied by $p$. We call a particle in the occupied neighbourhood of $p$ a \textit{neighbour} of $p$. Finally, we assume that a particle $p$ knows whether two adjacent nodes in $N(p)$ are occupied by one expanded particle or by two different particles. Notice that if $p$ is adjacent to only one endpoint of an expanded particle it does not know whether the neighbouring particle is expanded or contracted. We discuss the importance of this assumption in Section \ref{sec:movement-algo}. 

Within the Amoebot model, particles move using two main operations: \textit{expansion} to a neighbouring node and \textit{contraction} to the node occupied by the head of the particle. An expansion is only possible for a contracted particle and a contraction is only possible for an expanded particle. Contrary to Amoebot, here we assume that particles have two \textit{contract} operations: \textit{contract to head} and \textit{contract to tail}. The new \textit{contract to tail} operation does not alter the model significantly as it is equivalent to an expanded particle contracting to head, expanding in the opposite direction and contracting to head a second time. 

We say that a particle for which some condition is enabled is \textit{activable}. We call an activable particle that is chosen by the scheduler \textit{activated}. When a particle is activated, it detects which of its neighbouring nodes are occupied, whether any of those nodes are occupied by an expanded particle and, based on this information, it performs at most one contraction and at most one expansion, in this order. Operations performed during a single activation of the particle are called a \textit{move} and a move is assumed to be atomic. That is, when a particle is activated it performs all operations that are part of the move before becoming inactive and while a particle is active no other particle becomes activated. Notice that since a move consists of at most one contraction and at most one expansion, at least one of the nodes the particle occupied before its activation remains occupied by the same particle after its activation. Atomic moves are particularly important in this context, as they guarantee that the system remains connected after every \textit{move}. Without this assumption, the impossibility result of \cite{derakhshandeh2015leader} forbids movement within the self--stabilising context. 

We assume that all particles have a common sense of orientation (i.e., all particles agree on all directions), as shown in Figure \ref{fig:compass}.
\begin{figure}[t]
    \centering
    \begin{tikzpicture}[scale=.6]

        \draw[-Stealth] (0, 0) -- (60:1.5cm) node[label={[label distance=.1mm]60:\footnotesize 5}] {}; 
        \draw[-Stealth] (0, 0) -- (120:1.5cm) node[label={[label distance=.1mm]120:\footnotesize 4}] {}; 
        \draw[-Stealth] (0, 0) -- (180:1.5cm) node[label={[label distance=.1mm]180:\footnotesize 3}] {}; 
        \draw[-Stealth] (0, 0) -- (240:1.5cm) node[label={[label distance=.1mm]240:\footnotesize 2}] {}; 
        \draw[-Stealth] (0, 0) -- (300:1.5cm) node[label={[label distance=.1mm]300:\footnotesize 1}] {}; 
        \draw[-Stealth] (0, 0) -- (0:1.5cm) node[label={[label distance=.5mm]0:\small 0}] {};

    \end{tikzpicture}
    \caption{\centering Directions of particles}
    \label{fig:compass}
\end{figure}
We define a \textit{step} of the execution to be the time needed for a particle to become activated, execute a given algorithm once and become deactivated. We assume that particles are activated by a \textit{sequential unfair} scheduler. A sequential scheduler refers to the synchronicity of the system and guarantees that only one particle is active at any step and is a usual assumption within programmable matter systems (e.g., \cite{dufoulon2021efficient,emek2019deterministic,gastineau2019distributed}). An unfair scheduler concerns the fairness of activations and represents the weakest fairness condition per \cite{dubois2011taxonomy}: at each step, some activable particle is activated. In contrast to a fair scheduler, it does not require that every activable particle is eventually activated, and thus it is possible that any number of activable particles are perpetually ignored as long as at least one other activable particle exists in the system. 

The pseudocode is composed of a set of states which in turn are composed of a set of conditions and actions. Each line of the algorithm is of the form $c_i : a_1, \ldots, a_m$ where $c_i$ is the $i$-th condition and each $a_j \in \{a_1, \ldots, a_m\}$ is an action. Every time a particle is activated, it evaluates the conditions in its current state in order and if a condition is satisfied the particle performs the corresponding actions and finishes the round without evaluating the remaining conditions. 

\section{Silent Self--Stabilising Leader Election with Movement}\label{sec:movement-algo}
The algorithm we give follows the core idea presented in \cite{navarra2023asynchronous}. That is, the goal is to move particles that have a lower neighbour downwards if an empty lower node exists. However, in this work we also assume that any number of particles may initially be expanded and particles do not have 2--hop visibility. Since the labels of \textit{head} and \textit{tail} are usually part of the memory and we consider oblivious particles, every time an expanded particle is activated it assigns its head to be the end occupying the lower node if one exists, otherwise it assigns the head to be the end occupying the rightmost node. This, although not part of the Amoebot model, is not too restrictive as any particle for which the above orientation is not true could contract and expand in the opposite direction to achieve the described positioning. 

Informally, the algorithm is the following. A contracted particle that has exactly one lower neighbour (i.e., direction 1 or 2 of Figure \ref{fig:compass}) expands downwards without disconnecting its occupied neighbourhood. Any contracted particle that has either zero or two occupied lower neighbouring nodes and that has an upper neighbour to the right (i.e., in direction 5 of Figure \ref{fig:compass}), expands to the right (i.e., direction 0). An expanded particle contracts to its head as long as it does not disconnect its neighbours. If an expanded particle with at least one lower neighbour cannot contract without disconnecting its occupied neighbourhood, it tries to improve its position. This is done by remaining expanded and either moving one of its endpoints to a lower position that does not disconnect its neighbours or by moving the tail to a node that allows a different particle to move without disconnecting the system. An example of the former case is shown in Figure \ref{fig:expanded-cond-one-empty}. A central step in proving the correctness of our algorithm (i.e., Lemma \ref{lem:unique-leader}) is showing that:

\begin{framed}
    \centering When executing this algorithm there eventually exists exactly one particle that has no lower neighbours and no neighbour to the right (i.e., in directions 0, 1, 2 and 5). This particle is defined to be the leader.
\end{framed}

\begin{figure}[ht]
    \begin{subfigure}[b]{0.23\textwidth}
        \centering
        \includegraphics[scale=.6,page=1]{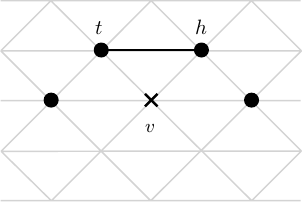}
        \caption{\centering}
    \end{subfigure}\hfill
    \begin{subfigure}[b]{0.23\textwidth}
        \centering
        \includegraphics[scale=.6,page=2]{expanded-cond-one-empty.pdf}
        \caption{\centering}
    \end{subfigure}\vspace{.3cm}
    \begin{subfigure}[b]{0.23\textwidth}
        \centering
        \includegraphics[scale=.6,page=3]{expanded-cond-one-empty.pdf}
        \caption{\centering}
    \end{subfigure}\hfill
    \begin{subfigure}[b]{0.23\textwidth}
        \centering
        \includegraphics[scale=.6,page=4]{expanded-cond-one-empty.pdf}
        \caption{\centering}
    \end{subfigure}
    \caption{An expanded particle (symbolised by two circles connected with a line) that cannot contract without disconnecting its neighbours (depicted as black circles). The tail is denoted by $t$ and the head by $h$. Nodes marked by \texttt{x} are assumed to be empty.}
    \label{fig:expanded-cond-one-empty}
\end{figure}

We now describe how a particle moves without disconnecting its neighbourhood. Let $p$ be an expanded particle and let $h$ and $t$ be the head and tail of $p$ respectively. Call $v$ the lower common neighbouring node of $h$ and $t$. If $p$ is expanded diagonally, call $w$ the higher common neighbour of $h,t$. If for every particle in $N(t)$ there exists a path to $h$ that only passes through occupied nodes in $N(p) \backslash \{t\}$, $p$ contracts to $h$ without disconnecting its neighbourhood. Otherwise, if $v$ is empty and for every particle in $N(t)$ there exists a path to $v$ that only passes through $v$ and occupied nodes in $N[h] \backslash \{t\}$, $p$ contracts to head and expands to $v$ so that $p$ occupies $v$ and the original node of $h$. An example of this case is shown in the first three subfigures of Figure \ref{fig:expanded-cond-one-empty}. If $p$ is expanded horizontally, $v$ is empty and for every particle in $N(h)$ there exists a path to $v$ that only passes through $v$ and occupied nodes in $N[t] \backslash \{h\}$, $p$ contracts to tail and expands to $v$ so that $p$ occupies $v$ and the original node of $t$. An example of this case is shown in the fourth subfigure of Figure \ref{fig:expanded-cond-one-empty}. When no such $v$ that preserves connectivity exists, if $p$ is expanded diagonally and $t$ is the lower common neighbour of a horizontally expanded particle, $p$ contracts to head and expands to $w$ so that $p$ occupies $w$ and the original node of $h$.

\subsection*{Pseudocode}

We use the following conditions in state \texttt{Expanded}:
\begin{enumerate}[E1]
    \item\label{cond:expanded-contract} Contracting to head does not disconnect neighbours. For example, the blue expanded particle in Figure \ref{fig:rules}.
    \item\label{cond:expand-lower-head} There exists a lower empty common neighbour, $v$, such that $N(p) \backslash \{t\} \cup \{v\}$ is connected. For example, the black square expanded particle in Figure \ref{fig:rules}.
    \item \label{cond:expand-lower-tail} The particle $p$ is horizontally expanded and there exists a lower empty node, $u$, that is neighbouring to the tail such that $N(p) \backslash \{h\} \cup \{u\}$ is connected. For example, the white circle expanded particle in Figure \ref{fig:rules}.
    \item\label{cond:expand-other-diagonal} The tail is the lower common neighbour of an expanded particle and the higher common neighbour of the head and tail, $w$, is empty. For example, the white square expanded particle in Figure \ref{fig:rules} and the pink expanded particle in Figure \ref{fig:pendulum}.
\end{enumerate}

We use the following conditions in state \texttt{Contracted}:
\begin{enumerate}[C1]
    \item \label{cond:contracted-expand} Out of the two lower neighbours one is occupied and one is empty. For example, the blue contracted particle in Figure \ref{fig:rules}.
    \item \label{cond:contracted-lowermost-right} Both lower neighbours are empty or both lower neighbours are occupied, there exists an occupied neighbour in direction 5 and the right neighbouring node (i.e., in direction 0) is empty. For example, the grey contracted particle in Figure \ref{fig:rules}.
\end{enumerate}
\begin{figure}[ht]
    \begin{subfigure}[t]{.49\textwidth}
        \centering
        \includegraphics[scale=.5,page=1]{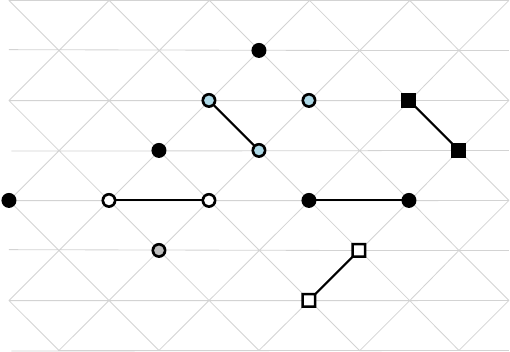}
        \caption{\centering}
        \label{fig:rules}
    \end{subfigure}\hfill
    \begin{subfigure}[t]{.49\textwidth}
        \centering
        \includegraphics[scale=.5,page=2]{rules.pdf}
        \caption{\centering}
        \label{fig:pendulum}
    \end{subfigure}
    \caption{(a) An example configuration containing particles that satisfy all conditions. The blue expanded particle satisfies \eref{cond:expanded-contract}, the black square expanded particle satisfies \eref{cond:expand-lower-head}, the white circle expanded particle satisfies \eref{cond:expand-lower-tail}, the white square expanded particle satisfies \eref{cond:expand-other-diagonal}, the blue circle contracted particle satisfies \cref{cond:contracted-expand} and the grey contracted particle satisfies \cref{cond:contracted-lowermost-right}. (b) A configuration where the only satisfied condition is \eref{cond:expand-other-diagonal} (for the pink expanded particle).}
    \label{fig:rules-and-pendulum}
\end{figure}

In the pseudocode, we write $h$ to denote the node occupied by the head of an expanded particle and $t$ to denote the node occupied by the tail of an expanded particle. Nodes $v,u,w$ are those defined in the conditions \eref{cond:expand-lower-head}, \eref{cond:expand-lower-tail} and \eref{cond:expand-other-diagonal} respectively. The comments in the pseudocode are the nodes occupied by the particle after its activation.

\begin{algorithm}[ht]
    \caption{Silent and Self--Stabilising Leader Election}
    \label{alg:LE-movement}

    \nonl\underline{In state \texttt{Expanded}:}\\
    \quad mark the lowest endpoint (or rightmost endpoint if both endpoints are in the \\ 
    \nonl\quad same row) as head and the remaining endpoint as tail\\
    
    \quad E\ref*{cond:expanded-contract}: contract to head, \texttt{Contracted}; \tcp*{occupies $h$\quad\hspace{1pt}}
    \quad E\ref*{cond:expand-lower-head}: contract to head, expand to $v$, \texttt{Expanded}; \tcp*{occupies $h,v$\hspace{2pt}}

    \quad E\ref*{cond:expand-lower-tail}: contract to tail, expand to $u$, \texttt{Expanded}; \tcp*{occupies $t,u$\hspace{3pt}}

    \quad E\ref*{cond:expand-other-diagonal}: contract to head, expand to $w$, \texttt{Expanded}; \tcp*{occupies $h,w$}

    \BlankLine
    \nonl\underline{In state \texttt{Contracted}:}\\

    \quad C\ref*{cond:contracted-expand}: expand to empty lower neighbour, \texttt{Expanded}; \\

    \quad C\ref*{cond:contracted-lowermost-right}: expand to empty right neighbour, \texttt{Expanded}; \\
\end{algorithm}

The ability of a particle to distinguish whether two of its neighbouring nodes are occupied by a single expanded particle or by two distinct particles is essential for the evaluation of \eref{cond:expand-other-diagonal}, which in turn allows progress in configurations such as the one depicted in Figure \ref{fig:pendulum}. If a particle satisfying \eref{cond:expand-other-diagonal} were to move without this distinction, an unfair scheduler could repeatedly activate only that particle, causing it to oscillate between two nodes indefinitely and preventing any progress.

Observe that if a particle satisfies \eref{cond:expanded-contract}, \eref{cond:expand-other-diagonal} or \cref{cond:contracted-lowermost-right} no endpoint of the particle occupies a lower node after its activation.  By definition, an expanded particle can only move one of its endpoints to a lower node through \eref{cond:expand-lower-head} or \eref{cond:expand-lower-tail} if the expanded particle has a lower neighbour and a contracted particle only expands to a lower node if it has a lower neighbour by definition of \cref{cond:contracted-expand}. Therefore:

\begin{observation} \label{obs:lower-neighbour}
A contracted particle $p$ that occupies some node $v$ can only occupy a node lower than $v$ after a move if $p$ has a lower neighbour when it is activated.
\end{observation}

\section{Proof of Correctness}\label{sec:correctness}
We show that the particle system remains connected and that there is an activable particle in every step or there is a unique particle that does not have a neighbour in directions 0,1,2 and 5. We call a configuration with no activable particle a \textit{final} configuration. By slight abuse of notation we refer to a node occupied by a particle $p$ as $p$. Throughout this section we use $h$ and $t$ to denote the head and tail of an expanded particle. When the particle we refer to is not clear from the context, we write $p_h,p_t$ instead of $h,t$ for the head and tail of a particle $p$. We call the node at direction $i$ of $h$ (resp. of $t$) $h_i$ (resp. $t_i$), regardless of whether it is empty or occupied. 

We will show that every final configuration has the following property. 
\begin{proposition}\label{prop:final}
    Every particle in a final configuration is in one of the following cases:
    \begin{enumerate}
        \item\label{case:contracted} Contracted and has either zero or two lower neighbours.
        \item Horizontally expanded, contracting disconnects its neighbours and has no lower neighbours.
    \end{enumerate}
\end{proposition}

We first establish the following lemma.

\begin{lemma}\label{lem:induction-step}
    Consider a particle $p$ in a final configuration $C$ s.t. for any $p'$ above $p$ or $p'$ at the same height but on the left of $p$, Proposition \ref{prop:final} holds. Then Proposition \ref{prop:final} holds for $p$.
\end{lemma}
\begin{proof}
    It is easy to see that any contracted particle with one empty and one occupied lower neighbouring node is activable due to \cref{cond:contracted-expand}. Hence we immediately get that if $p$ is contracted in a final configuration $C$ it has either zero or two lower neighbours. We now show that if $p$ has a lower neighbour in $C$, $p$ is contracted. We split the proof into two cases based on whether $p$ is expanded horizontally or diagonally and we show that either $p$ is expanded horizontally with no lower neighbours in $C$ or we get a contradiction to $C$ being final.

    \begin{case}\label{case:exp-horizontally}
        Particle $p$ is expanded horizontally.
    \end{case}
    If the occupied neighbourhood of $p$ is connected, $p$ is activable due to \eref{cond:expanded-contract} and $C$ is not final. So let us assume that $N(p)$ is not connected and consider cases based on $N(p)$. 

    \begin{subcase}\label{case:exp-horizontally-t4-empty}
        Neighbour $t_3$ is empty.
    \end{subcase}
    Observe that in this case $t_4$ must be empty since by assumption the lemma holds for $t_4$ (i.e., a particle at $t_4$ cannot have only one lower neighbour regardless of whether it is occupied by an expanded or a contracted particle). The only remaining neighbour of $t$ that is not a neighbour of $h$ is $t_2$. We consider two cases based on whether $t_2$ is occupied. This is shown in Figure \ref{fig:ind-step-horizontal-t2}.
    \begin{figure}[ht]
        \begin{subfigure}[t]{0.49\textwidth}
            \centering
            \includegraphics[scale=.65,page=5]{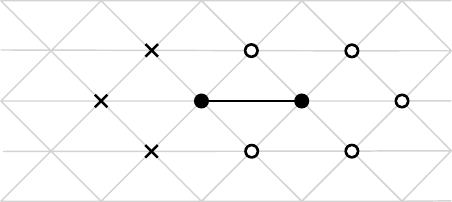}
            \caption{\centering $t_2$ empty}
            \label{fig:ind-step-horizontal-t2-empty}
        \end{subfigure}
        \begin{subfigure}[t]{0.49\textwidth}
            \centering
            \includegraphics[scale=.65,page=6]{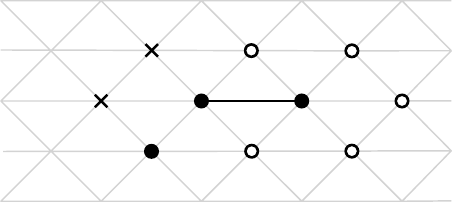}
            \caption{\centering $t_2$ occupied}
            \label{fig:ind-step-horizontal-t2-occ}
        \end{subfigure}
        \caption{Black circles represent particles. Two black circles connected with a line represent an expanded particle. White circles can be empty or occupied by particles. Nodes marked by \texttt{x} are assumed to be empty.}
        \label{fig:ind-step-horizontal-t2}
    \end{figure}
    Let us begin by assuming $t_2$ is empty. Then $p$ is activable by \eref{cond:expanded-contract} as all occupied neighbours of $t$ are also neighbours of $h$. So let us assume that $t_2$ is occupied. Then if $t_1$ is empty, $p$ is activable by \eref{cond:expand-lower-head}. The resulting configuration is connected as all neighbours of $h$ are still neighbouring to $p$ and $t_2$ is neighbouring to $t_1$ that is now occupied by $p$. So let us assume that $t_1$ is occupied. Then since $t_2$ is neighbouring to $t_1$, all neighbours of $p$ are connected to $h$ and $p$ is activable by \eref{cond:expanded-contract}. 

    \begin{subcase}\label{case:exp-horizontally-t4-occ}
        Neighbour $t_3$ is occupied.
    \end{subcase}
    \begin{figure}[ht]
        \begin{subfigure}[t]{0.32\textwidth}
            \centering
            \includegraphics[scale=.6,page=3]{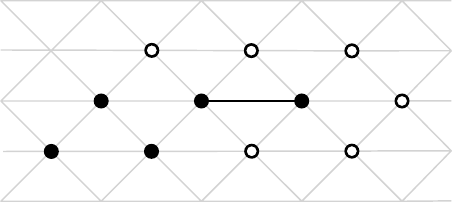}
            \caption{\centering $t_3$ contracted}
            \label{fig:ind-step-horizontal-cont-2}
        \end{subfigure}\hfill
        \begin{subfigure}[t]{0.32\textwidth}
            \centering
            \includegraphics[scale=.6,page=2]{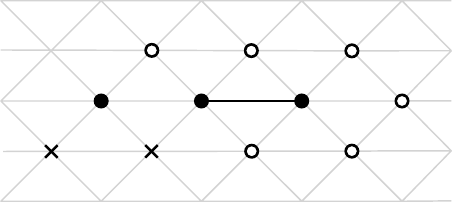}
            \caption{\centering $t_3$ contracted}
            \label{fig:ind-step-horizontal-cont-0}
        \end{subfigure} \hfill
        \begin{subfigure}[t]{.32\textwidth}
            \centering
            \includegraphics[scale=.6,page=4]{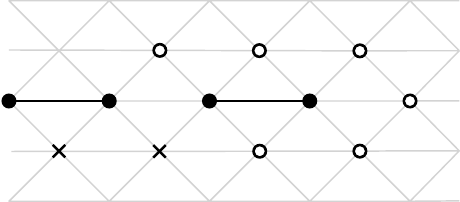}
            \caption{\centering $t_3$ expanded}
            \label{fig:ind-step-horizontal-exp}
        \end{subfigure}
        \caption{Black circles represent particles. Two black circles connected with a line represent an expanded particle. White circles can be empty or occupied by particles. Nodes marked by \texttt{x} are assumed to be empty.}
    \end{figure}
      
    Let us first assume that $t_2$ is occupied (i.e., Figure \ref{fig:ind-step-horizontal-cont-2}). If $t_1$ is occupied, $p$ is activable due to \eref{cond:expanded-contract}. The resulting configuration is connected due to the neighbourhood of $t$ being connected and $h$ neighbouring to $t_1$ which is also neighbouring to the occupied neighbour of $t$, $t_2$. If $t_1$ is empty, $p$ is activable due to \eref{cond:expand-lower-head}. This is a valid configuration as $h$ is still occupied by $p$ so neighbours of $h$ remain connected, the neighbourhood of $t$ is connected and $t_1$ that is now occupied by $p$ is neighbouring to the occupied neighbour of $t$, $t_2$. 
    
    Let us now assume that $t_3$ does not have lower neighbours, or equivalently that $t_2$ is not occupied since the lemma holds for $t_3$. Notice that since $p$ cannot know whether $t_3$ is contracted or expanded even though two cases are possible for $t_3$ as shown in Figures \ref{fig:ind-step-horizontal-cont-0} and \ref{fig:ind-step-horizontal-exp} the analysis for both is common. If $h_2 = t_1$ is occupied and $h_1$ is empty, we consider two cases based on whether the particle occupying $h_2$, $q$, is contracted or expanded. If $q$ is contracted, $q$ is activable from \cref{cond:contracted-expand} if it has a unique lower neighbour or from \cref{cond:contracted-lowermost-right} due to $h$, otherwise. As a contracted particle expanding cannot disconnect the system the resulting configuration is connected. If $q$ is expanded, it must be expanded diagonally and $q$ is activable from \eref{cond:expand-other-diagonal}. Since both $t_2$ and $h_1$ are empty in this case $q$ moving does not disconnect its neighbourhood as the only neighbour of $q_t$ that is not a neighbour of $q_h$ is $p$ which is still a neighbour of $q$ in the new configuration. Let us now assume that $h_1$ is occupied and $h_2$ is empty. Then, if $h_5$ is occupied by the lemma statement $h_0$ must also be occupied. Therefore, for every neighbour of $h$ that is not a neighbour of $t$ there is a path to $h_1$ that only passes through occupied nodes in $N(h)$. In this case, $p$ is either activable by \eref{cond:expand-lower-tail} and in the resulting configuration occupies $t$ and $h_2$ or by \eref{cond:expanded-contract}. In the former case, as all neighbours of $t$ are still neighbouring to $t$, the neighbours of $h$ are connected and $h_2$ is neighbouring to the occupied neighbour of $h$, $h_1$ the resulting configuration is connected. The latter case is possible if $t_4$ and $t_5 = h_4$ are also occupied and the resulting configuration is connected since $t_3,t_4,t_5$ are connected and $t_5$ is neighbouring to $h$. Suppose now that $h_1$ and $h_2$ are both occupied. Then $p$ is activable either by \eref{cond:expand-lower-tail} and in the resulting configuration occupies $t$ and $t_2$ or by \eref{cond:expanded-contract}. In the former case, if $h_5$ is occupied, $h_0$ must also be occupied by the lemma statement. Since $h_0$ is neighbouring to $h_1$, for every neighbour of $h$ that is not a neighbour of $t$ there is a path to $t_1=h_2$ and this path only passes through occupied nodes in $N(h)$. Since $t$ is neighbouring to $t_1$ and is additionally connected to all its neighbours, the resulting configuration is connected. The latter case is possible if $t_5 = h_4$ and $t_4$ are both occupied and the resulting configuration is connected since $t_3,t_4,t_5$ are connected and $t_5$ is neighbouring to $h$. Finally, if both $h_2$ and $h_1$ are empty Proposition \ref{prop:final} holds. 

    Therefore, if $p$ is expanded horizontally, Proposition \ref{prop:final} holds.

    \begin{case}\label{case:exp-diagonally}
        Particle $p$ is expanded diagonally.
    \end{case}
    We consider the case of $p$ being expanded towards direction 1 (i.e., $h_4 = t$ and $h = t_1$) and the other case is symmetric. If the occupied neighbourhood of $p$ is connected without $t$, $p$ is activable due to \eref{cond:expanded-contract}. So we assume that the occupied neighbourhood of $p$ is not connected without $t$. The only neighbours of $t$ that are not neighbours of $h$ are $t_3$, $t_4$ and $t_5$. By assumption, if $t_4$ is occupied, $t_3$ must be occupied and if $t_5$ is occupied, $t_0$ must be occupied. Hence we do not need to consider $t_4$ separately to $t_3$ or $t_5$ separately to $t_0$. Furthermore, $t_0$ is neighbouring to $h$ so if $t_5$ is occupied it is always connected to the neighbourhood of $h$. Therefore, the only neighbours of $p$ that we need to consider are $t_2,t_3$. If $t_3$ is not occupied or if $t_3$ is occupied and $t_2$ is occupied, $p$ is immediately activable by \eref{cond:expanded-contract}. The resulting configuration is valid as in either case if $t_2$ is occupied, it is a neighbour of $h$. So let us suppose that $t_3$ is occupied and $t_2$ is not occupied. Then $p$ is activable by \eref{cond:expand-lower-head}. The resulting configuration is connected as the neighbours of $t$ are connected and $t$ now occupies $t_2$ that is neighbouring to the occupied node $t_3$.  

    From Case \ref{case:exp-horizontally} and Case \ref{case:exp-diagonally} if $p$ is expanded, either it has no lower neighbour as is the second case of Proposition \ref{prop:final} or $C$ is not final which is a contradiction.
\end{proof}

We can now prove Proposition \ref{prop:final}.

\begin{proof}[Proof of Proposition \ref{prop:final}]
    Suppose now that Proposition \ref{prop:final} does not hold for all particles. Let $q$ be the topmost leftmost particle among the particles for which Proposition \ref{prop:final} does not hold. Then for all particles above $q$ and in the same height but on the left of $q$ the property holds. From Lemma \ref{lem:induction-step}, Proposition \ref{prop:final} then also holds for $q$, which is a contradiction. Hence, the proposition holds.
\end{proof}

We now show that the final configuration has a unique leader. 

\begin{lemma} \label{lem:unique-leader}
    A final configuration contains exactly one leader, that is, a particle that does not have any neighbours in directions 0,1,2 and 5.
\end{lemma}
\begin{proof}
    Suppose a final configuration does not have a leader. Consider the lowest row of particles and take the rightmost particle in this row, $p$. By definition, $p$ does not have neighbours in directions 0,1 and 2. Hence, $p$ must have a neighbour in direction 5, $p_5$. From Proposition \ref{prop:final}, $p_5$ must have a second lower neighbour, but by assumption $p_0$ (i.e., the neighbour of $p$ in direction 0) is empty. So this scenario is not possible and $p$ is a leader in the final configuration. 

    Let us now suppose that there exist at least two leaders in the final configuration. Call the leader that is not globally lowermost rightmost $p'$. Then there exists a shortest path from $p'$ to $p$ that minimises the number of particles in the highest row of the path. We denote that path $p' \leadsto p$. Take $q$ to be the first particle on that path that has a lower neighbour. From Proposition \ref{prop:final}, since $q$ has a lower neighbour, it must be contracted and have a second lower neighbour. Notice that $q$ is either at the same row as $p'$ or higher. 
    
    Let us first assume that $q$ is in the same row as $p'$. Then the predecessor of $q$ on $p' \leadsto q$ cannot be the neighbour of $q$ in direction 0 or in direction 3 since both of those nodes share a lower common node with $q$, and $q$ would not be the first particle in $p' \leadsto p$ that has a lower neighbour. Furthermore, the predecessor of $q$ on $p' \leadsto p$ cannot be in direction 5 or 4 of $q$ as $q$ is a lower neighbour of those nodes and $q$ would not be the first particle on $p' \leadsto p$ to have a lower neighbour. Hence the predecessor of $q$ on the path must be one of the neighbours of $q$ in direction 1 or 2. However, both of those neighbours are lower than $p'$ and as a result $q$ would not be the first particle in the path to have a lower neighbour. Therefore, this case is not possible. 
    
    Let us now suppose that $q$ is at some row that is higher than $p'$. Since $p$ is a globally lowermost rightmost node $p' \leadsto p$ must eventually move down again. Take $q'$ to be the first particle of $p' \leadsto p$ that is at the highest row of the path. The predecessor of $q'$ on the path cannot be one of the neighbours of $q'$ in directions 0,3,4 or 5 otherwise $q'$ would not be the first particle to be in the highest row of the path. Therefore, the predecessor of $q'$ is one of the neighbours of $q'$ in directions 1 or 2. Notice that from Proposition \ref{prop:final} since one of the lower neighbours of $q'$ is occupied, both neighbours in directions 1 and 2 must be occupied. Let us now consider the successor of $q'$ in $p' \leadsto p$. Since $q'$ is at the highest row of the path, its successor on the path cannot be in direction 4 or 5. Without loss of generality, let us assume that the predecessor of $q'$ on the path is in direction 1 and the case for the predecessor being in direction 2 is symmetric. If the successor of $q'$ is in direction $j \in \{0,2\}$, $p' \leadsto q'_1,q'_j \leadsto p$ is a shorter path than $p' \leadsto q'_1,q',q'_j \leadsto p$. This contradicts $p' \leadsto p$ being a shortest path and as a result this case is not possible. Finally, suppose the successor is in direction 3. Then since $q'_2$ is occupied, $p' \leadsto q'_1, q'_2, q'_3 \leadsto p$ is a path with fewer particles in the highest row that does not include $q'$, a contradiction. Therefore, this case is not possible either.

    Consequently, a final configuration has exactly one leader.
\end{proof}

\begin{figure}[ht]
    \begin{subfigure}[t]{0.32\textwidth}
        \centering
        \includegraphics[scale=.45,page=1]{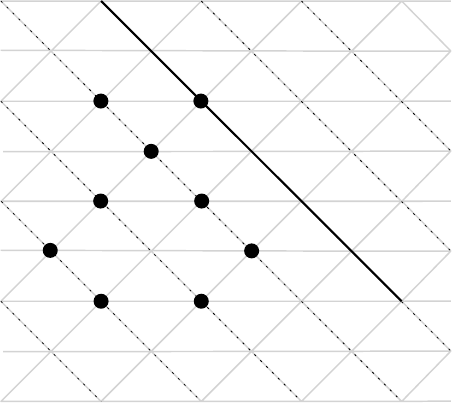}
        \caption{\centering}
        \label{fig:horizontal-axes}
    \end{subfigure}
    \begin{subfigure}[t]{0.34\textwidth}
        \centering
        \includegraphics[scale=.45,page=3]{progress.pdf}
        \caption{\centering}
        \label{fig:horizontal-dist}
    \end{subfigure}
    \begin{subfigure}[t]{0.3\textwidth}
        \centering
        \includegraphics[scale=.45,page=4]{progress.pdf}
        \caption{\centering}
        \label{fig:blocking}
    \end{subfigure}
    \caption{(a) Dotted lines represent all north-west to south-east diagonals on the grid and the black line represents the diagonal boundary. (b) The diagonal and horizontal boundary are marked by black lines, the particle represented as a black square is the particle that defines the diagonal boundary, an example of the distance to the diagonal boundary is shown for the two particles marked as white squares and an example of the distance to the horizontal boundary is shown for the particles marked as blue squares. (c) An example of a blocking particle marked by square endpoints.}
    \label{fig:progress}
\end{figure}

Finally, we show that a final configuration is reached from any starting configuration and for any order of activations. We call the lowest row that contains a particle the \textit{horizontal boundary} of the system. We define the \textit{vertical position} of a particle to be the head's distance to the horizontal boundary. For a contracted particle, the vertical position is the distance of the particle to the horizontal boundary. Consider all north-west to south-east diagonals of the grid (e.g., diagonals marked by a dotted line in Figure \ref{fig:horizontal-axes}). Call \textit{diagonal boundary} the rightmost such diagonal that contains a particle (e.g., the diagonal marked in Figures \ref{fig:horizontal-axes} and \ref{fig:horizontal-dist}). Observe that this particle is not necessarily a globally rightmost particle in the system. The \textit{horizontal position} of a particle is the horizontal distance of the particle to the diagonal boundary (e.g., Figure \ref{fig:horizontal-dist}). For a contracted particle the horizontal distance is the distance of the particle to the diagonal boundary and for an expanded particle the horizontal distance is the distance of the head to the diagonal boundary. Finally we write \textit{blocking particle} to denote a diagonally expanded particle whose tail is neighbouring to both the head and tail of an expanded particle (e.g., Figure \ref{fig:blocking}). 

We define \textit{progress} to be a decrease on one of the following criteria, evaluated in order. Once a criterion is satisfied, the remaining criteria are not evaluated.
\begin{enumerate}[P1]
    \item\label{prog:vertical} The sum of the vertical positions of the particles.
    \item\label{prog:horizontal} The sum of the horizontal positions of the particles.
    \item\label{prog:diag} The number of diagonally expanded particles.
    \item\label{prog:blocking} The number of blocking particles.
    \item\label{prog:horizontal-expanded} The number of horizontally expanded particles.
\end{enumerate}

The convergence argument below relies on invariants guaranteed by the movement rules defined in Algorithm \ref{alg:LE-movement}. In particular, particles only move to a lower row if a lower neighbour exists, do not cause the configuration to become disconnected and never move in directions 4 or 5. 

\begin{lemma}\label{lem:progress}
    Starting from any arbitrary initial configuration any sequential unfair execution of Algorithm \ref{alg:LE-movement} eventually reaches a final configuration.
\end{lemma}
\begin{proof}
    We prove this lemma by first showing that if an activable particle moves there is progress. We consider all cases of activable particles. 
    \begin{itemize}
        \item A particle that satisfies \cref{cond:contracted-expand}: If a particle expands diagonally, by construction of the algorithm the particle expands to a lower node and as a result \pref{prog:vertical} is decreased.
        \item A particle that satisfies \cref{cond:contracted-lowermost-right}: In this case a contracted particle expands towards direction 0, decreasing the horizontal distance of the particle to the diagonal boundary, that is, decreasing \pref{prog:horizontal}. Furthermore, \cref{cond:contracted-lowermost-right} only applies to contracted particles moving to an empty neighbour in the same row hence \pref{prog:vertical} does not change. 
        \item A particle that satisfies \eref{cond:expanded-contract}: There are two possibilities for this case: a diagonally expanded particle contracts to a lower node or a horizontally expanded particle contracts to the right. In the former case \pref{prog:diag} is reduced. Additionally, since the vertical and horizontal positions of a particle are only determined by the position of the head, which in this case does not move, \pref{prog:vertical} and \pref{prog:horizontal} do not change. In the latter case, \pref{prog:blocking} or \pref{prog:horizontal-expanded} is reduced, depending on the configuration. Criteria \pref{prog:vertical}, \pref{prog:diag} are affected by particles being or becoming diagonally expanded and as a result do not change. Criterion \pref{prog:horizontal} is not reduced as a particle satisfying \eref{cond:expanded-contract} contracts to the node occupied by its head and as a result does not change its distance to the diagonal boundary. 
        \item A particle that satisfies \eref{cond:expand-lower-head} or \eref{cond:expand-lower-tail}: Once again, this condition can either be satisfied by a horizontally expanded or a diagonally expanded particle. If the particle is diagonally expanded it must satisfy \eref{cond:expand-lower-head}. In this case, the particle becomes horizontally expanded thus reducing \pref{prog:diag}. Notice that the position of the head does not change hence \pref{prog:vertical} and \pref{prog:horizontal} do not change. If the particle is expanded horizontally, the particle becomes diagonally expanded to a lower node decreasing \pref{prog:vertical}.
        \item A particle that satisfies \eref{cond:expand-other-diagonal}: In this case a diagonally expanded particle moves but remains diagonally expanded. The action corresponding to \eref{cond:expand-other-diagonal} reduces \pref{prog:blocking}. Furthermore, since the position of the head does not change, \pref{prog:vertical} and \pref{prog:horizontal} do not change and since the particle remains diagonally expanded, \pref{prog:diag} does not change.  
    \end{itemize}
    Therefore as long as there is an activable particle in the system,  progress is ensured. We also need to show that eventually there do not exist activable particles in the system. We prove this by showing that no particle moves lower than the horizontal boundary and no particle moves to the right of the diagonal boundary. From Observation \ref{obs:lower-neighbour}, a particle only moves to a lower row if it has a lower neighbour. Hence particles in the lowest row of the system never move downwards. By construction of the algorithm particles move to distance at most one during an activation so particles that are not on the diagonal boundary cannot cross it in one activation. So let us suppose that there exists a particle $p$ that is the first particle that can cross the diagonal boundary during an execution of Algorithm \ref{alg:LE-movement}. From the structure of the grid, $p$ cannot cross the diagonal boundary when expanding or contracting to a lower neighbour. The only way $p$ can cross the diagonal boundary is by moving to direction 5 or 0. In Algorithm \ref{alg:LE-movement}, particles never move to direction 5. Furthermore, in Algorithm \ref{alg:LE-movement}, a particle only moves to direction 0 only if it has a neighbour in direction 5 through condition \cref{cond:contracted-lowermost-right}. However, a neighbour in direction 5 is beyond the diagonal boundary. As we have assumed that $p$ is the first particle that crosses the diagonal boundary, this is a contradiction. 
    
    Since the positions that the particle system can occupy are bounded below and to the right and particles never move in any other direction, eventually no particle is activable.
\end{proof}

From Proposition \ref{prop:final} and Lemmas \ref{lem:unique-leader} and \ref{lem:progress} we get:

\begin{theorem}
    Algorithm \ref{alg:LE-movement} eventually reaches a final configuration that contains exactly one particle without neighbours in directions 0,1,2 and 5, in a silent and self--stabilising manner.
\end{theorem}

\begin{observation}\label{obs:marked-leader}
    Algorithm \ref{alg:LE-movement} can be trivially transformed to an algorithm where the leader is explicitly marked.
\end{observation}
One way to perform this transformation is by adding a single bit of memory to each particle that can be set to 0 to denote that a particle cannot be a leader based on local conditions and set to 1 otherwise. The correctness of our algorithm implies that eventually only one particle has its variable set to 1. 

\section{Conclusion}
We presented the first self--stabilising result in programmable matter for the case of an unfair scheduler and for the case of systems containing holes. As we now know that even in a self--stabilising setting, using the movement capabilities of particles can lead to solving problems that appear very challenging in the stationary setting, a natural direction for future work is studying what other problems would benefit from particle movement in a self--stabilising context. Another interesting direction would be to study self--stabilising leader election in connected systems under weaker assumptions. For example, one could consider substituting the agreement on all directions with a common sense of rotational orientation (i.e., chirality). A starting point could be to attempt to compact the system (i.e., ``fill in'' the holes). However, this setting presents additional challenges as the particles cannot locally recognise the inside and the outside of the configuration. Hence, oblivious particles may not be able to use such a method. In the non self--stabilising context there exist methods to calculate this information (e.g., \cite{dufoulon2021efficient,emek2019deterministic}) with constant memory but it is not clear whether a self--stabilising counterpart for these methods can be designed. A different example is considering a synchronous or asynchronous scheduler instead of a sequential one. Although Algorithm \ref{alg:LE-movement} works if the activated particles are not neighbouring, it is easy to see that this is not the case for neighbouring activable particles (e.g., Figure \ref{fig:counterexample}).

\begin{figure}
    \centering
    \begin{subfigure}[b]{0.23\textwidth}
        \centering
        \includegraphics[page=1,scale=.6]{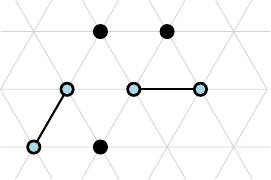}
        \caption{}
    \end{subfigure}
    \begin{subfigure}[b]{0.23\textwidth}
        \centering
        \includegraphics[page=2,scale=.6]{counterexample-not-seq.pdf}
        \caption{}
    \end{subfigure}
    \begin{subfigure}[b]{0.23\textwidth}
        \centering
        \includegraphics[page=3,scale=.6]{counterexample-not-seq.pdf}
        \caption{}
    \end{subfigure}
    \begin{subfigure}[b]{0.23\textwidth}
        \centering
        \includegraphics[page=4,scale=.6]{counterexample-not-seq.pdf}
        \caption{}
    \end{subfigure}
    \caption{Two examples of neighbouring particles activated simultaneously disconnecting the system. Suppose the blue expanded particles of Figure (a) (resp. Figure (c)) are simultaneously activated. Then \eref{cond:expanded-contract} is satisfied for both particles (resp. \eref{cond:expanded-contract} and \eref{cond:expand-other-diagonal} are satisfied for each expanded particle) and the resulting configuration becomes disconnected as shown in Figure (b) (resp. Figure (d)).}
    \label{fig:counterexample}
\end{figure}

The most interesting open question is whether it is possible to design a stationary self--stabilising leader election algorithm using constant memory for arbitrary connected configurations. We conjecture that the design of such an algorithm is impossible, even if particles agree on all directions of the grid.  

{
\fontsize{9pt}{\baselineskip} \selectfont
\subsubsection*{Acknowledgments.} JC was partially funded by
ANR project MIMETIQUE “Mineurs métriques” (ANR-25-CE48-4089-01). MK was supported by the DFG Project SCHE 1592/10-1.

\subsubsection*{Disclosure of Interests.} The authors have no competing interests to declare that are relevant to the content of this article.
}

\bibliographystyle{plainurl}
\bibliography{biblio}

\end{document}